%% file: root.tex
\documentclass{ifacconf}

\usepackage{graphicx}      
\usepackage{natbib}        

\input{command.tex}
\begin{document}
\begin{frontmatter}

\title{An Individual-Delay-Reflected Generalized Consensus Analysis for Multi-Agent Systems with Heterogeneous Time-Varying Delays} 

\author[First]{Hye Jin Lee\thanksref{equal}} 
\author[Second]{Ho Sub Lee\thanksref{equal}} 
\author[First]{PooGyeon Park\thanksref{cor}}
\thanks[cor]{Corresponding author}
\thanks[equal]{These authors contributed equally to this work.}
\address[First]{Department of Electrical Engineering,
	POSTECH, Pohang, Korea, (e-mail: hyejin2@postech.ac.kr, ppg@postech.ac.kr)}
\address[Second]{Samsung Electronics, Suwon, Korea (e-mail: hoss.lee@samsung.com)}

\copyright\ 2026 the authors. This work has been accepted to IFAC for publication under a Creative Commons Licence CC-BY-NC-ND.

\begin{abstract}                
In multi-agent systems, heterogeneous time delays exist for all agents because of the difference in communication environments. Therefore, the consensus analysis of a system considering a homogeneous time-varying delay among all agents results in conservatism. In this study, an individual-delay-reflected generalized consensus is proposed for multi-agent systems with heterogeneous time-varying delays with various bounds. 
To reflect heterogeneous time-varying delays, the proposed Lyapunov–Krasovskii functional is constructed by dividing the integral term into intervals containing heterogeneous delays and considering augmented vectors with delay states and integral states. Furthermore, by adding zero equality conditions, conservatism is reduced. $N$-dependent generalized integral inequality is used to allow the user to adjust the computational complexity. Numerical examples demonstrate a reduction in conservatism with the proposed consensus criterion.
\end{abstract}


\end{frontmatter}

\section{Introduction}
\input{txts/sec1}
\section{Problem Statement}
\input{txts/sec2}

\section{Main Result}
\input{txts/sec3}
\section{Numerical Example}
\input{txts/sec4}

\section{Conclusion}
In this study, an individual-delay-reflected consensus analysis is proposed for MASs with heterogeneous 
time-varying delays. Unlike existing works, this study considers heterogeneous time-varying delays with distinct 
minimum and maximum bounds for each agent. The proposed consensus criterion accounts for these delays by dividing 
the integral term in the LKF. Additionally, by constructing the LKF with individual Lyapunov matrices for each agent, 
incorporating augmented states that include both delays and integral states, and applying zero-equality conditions, 
conservatism is reduced. The consensus analysis also employs an $N$-dependent generalized integral inequality, allowing 
for flexibility in adjusting the computational complexity by tuning $N$. Numerical examples demonstrate improvement 
in the proposed consensus criterion to reduce conservatism. 
The proposed framework provides a quantitative criterion for determining the maximum tolerable delay in networked 
multi-agent systems, offering a practical guideline for delay-aware consensus design.

\begin{ack}
This work was supported by the National Research Foundation of Korea (NRF) grant funded by the Korea government (MSIT) (RS-2025-00513307 and RS-2025-16065922).
\end{ack}

\bibliography{reference}            
                                                                 
\end{document}

%% file: command.tex
\usepackage{graphics}
\usepackage{graphicx}
\usepackage{lipsum} 
\usepackage{amsmath}
\usepackage{amssymb}  
\usepackage{array,multirow}

\newtheorem{theorem}{Theorem}
\newtheorem{lemma}{Lemma}
\newtheorem{corollary}{Corollary}
\newtheorem{remark}{Remark}
\newtheorem{example}{Example}

\newcommand{\beq}{\begin{equation}}
\newcommand{\eeq}{\end{equation}}
\newcommand{\beqn}{\begin{equation*}}
\newcommand{\eeqn}{\end{equation*}}
\newcommand{\beqset}{\begin{equation}\begin{aligned}}
\newcommand{\eeqset}{\end{aligned}\end{equation}}
\newcommand{\beqnset}{\begin{equation*}\begin{aligned}}
\newcommand{\eeqnset}{\end{aligned}\end{equation*}}
\newcommand{\bag}{\begin{align}}
\newcommand{\eag}{\end{align}}
\newcommand{\bagn}{\begin{align*}}
\newcommand{\eagn}{\end{align*}}
\newcommand{\baged}{\begin{aligned}}
\newcommand{\eaged}{\end{aligned}}
\newcommand{\bb}{\begin{bmatrix}}
\newcommand{\eb}{\end{bmatrix}}
\newcommand{\barr}{\left[\begin{array}}
\newcommand{\earr}{\end{array}\right]}

%% file: txts/sec1.tex
Over the past few years, multi-agent systems (MASs) have attracted considerable attention due to their applications in fields such as unmanned aerial vehicles (UAVs), unmanned surface vehicles (USVs), electric vehicles, and mobile robots \cite{afrazi2025enhanced,lee2025graph,gao2022safety,hua2025multi,li2024multi}. 
In MAS research, the consensus problem that designs protocols that ensure all agents achieve agreement is a critical research topic that has motivated extensive studies from various perspectives, including research on hybrid MASs \cite{li2024hybrid, feng2024adaptive}, nonlinear MASs \cite{zhao2024observer, xue2023input, afrazi2025enhanced, yang2024nonsingular}, event triggered consensus \cite{hu2024event, ye2025adaptive, xu2024event}, and sampled data consensus \cite{wu2024sampled, zeng2023weighted}.

In practice, time delays are widely observed during communications among agents, which typically cause poor performance, oscillation, or even instability of the system.
Thus, extensive studies have been conducted on the consensus problem in MASs with various time delays including fixed, homogeneous, and heterogeneous delays \cite{vu2024matrix, wen2024output}.
The primary challenge in a time delay system is the stability analysis to determine the maximum allowable delay.
For instance, \cite{liu2021pseudo} establishes a pseudo-predictor feedback protocol for MASs with fixed homogeneous time delays.
A hybrid impulsive control protocol for nonlinear fuzzy MASs with fixed homogeneous time delays is presented in \cite{you2023control}.
Additionally, consensus protocols for linear MASs with fixed input and output time delays are studied in \cite{zhang2023leader}.
However, real systems are influenced by dynamic environmental conditions such as signal interference and network congestion and fixed delays cannot adequately capture these effects. 
Consequently, existing studies \cite{tang2023edge, li2023delay} have focused on consensus for MASs with homogeneous time-varying delays for all agents. 
This assumption entails a significant limitation because it implies that delays for all agents change uniformly and simultaneously. 
Therefore, to enhance the realism of the analysis, it is essential to consider heterogeneous delays. Thus, this study focuses on heterogeneous delays.

Heterogeneous time delays refer to the presence of individual delays for each agent, arising from differences in the communication environment (e.g., the state of the medium and distance) and variations in the processing speeds of individual agents.
Several studies have addressed the consensus problem in MASs under heterogeneous time-varying delays. For example, \cite{savino2015conditions} presents conditions for consensus in MASs with heterogeneous delays and switching topology, 
while \cite{deng2020dynamic} investigates an event-triggered approach under similar delay settings. In addition, \cite{jiang2022multi} examines consensus for MASs with heterogeneous time-varying input and communication delays. 
Notably, these studies assume that all delays have a common bound, which does not fully account for heterogeneous time-varying delays. 

Motivated by the aforementioned factors, this study proposes a generalized consensus analysis for MASs with heterogeneous time-varying delays, considering different bounds for each delay to ensure a higher allowable delay.  
As a result, it addresses a more realistic problem for scenarios that better reflect practical conditions and follows the general form of previous heterogeneous studies.  
In particular, the Lyapunov-Krasovskii functional (LKF) is constructed using the delay partitioning method and individual Lyapunov matrices to account for the specific delay of each agent.
Furthermore, an augmented state incorporating delay and integral states with heterogeneous time-varying delays is introduced to account for the correlated effects, and an $N$-dependent generalized integral inequality is employed.

$\mbox{~~} Notations:$ 
Throughout this paper,  $P>0$ denotes that $P$ is a positive definite matrix; $1_n$, $0_n$, and $I_n$ indicate n-dimensional column vectors of ones, zeros, and the $n \times n$ identity matrix, respectively; 
$e_i$ denotes a column vector with 1 at the $i$-th element and 0 elsewhere.
${\bf Sym}(X)$ denotes that $X+X^T$; $\otimes$ means the Kronecker product; 
The binomial coefficient is represented as ${p \choose q}={p!\over q!(p-q)!}$; 

%% file: txts/sec2.tex
\subsection{Graph theory}
Let $\mathcal{G}(\mathcal{V},\mathcal{E},\mathcal{A})$ be a directed graph, where $\mathcal{V}=\{v_1,\cdots,v_N\}$ is a set of $N$ vertices, $\mathcal{E}$ is a set of edges whose element pair $e_{ij}=(v_i,v_j)$ belongs to the set $\mathcal{E}$, and $\mathcal{A}=[a_{ij}]_{N \times N}$ is an adjacency matrix. If $e_{ij}$ is not zero, $a_{ij}=1$, and $a_{ij}=0$ otherwise. The Laplacian matrix is defined as $L = [l_{ij}]_{N \times N}$, where element $l_{ij} = -a_{ij}$ if $i \neq j$ and $l_{ii}=-\sum_{j=1,~j\neq i}^N l_{ij}$. 

\subsection{System description}
Let us consider the $k$th agent system with heterogeneous time-varying delays.
\beq\label{xi_sys}
\baged
    \dot x_k(t) &= Ax_k(t)+Bu_k(t-\tau_{k}(t)),  \\
    u_k(t)  &= -\sum_{l=1,l \neq k}^m a_{kl}K(x_k(t)-x_l(t)),
\eaged\eeq
where $x_k(t) \in \mathbb R^n$ and $u_k(t) \in \mathbb R^r$ denote the state vector and control input vector of the $k$th agent for $k \in\{1,\cdots,m\}$, respectively, and $m$ indicates the total number of agents. 
Here, $\tau_{k}(t)$ represents the heterogeneous time-varying delay associated with the $k$th agent and satisfies $0\leq\tau_k^l\leq\tau_{k}(t)\leq\tau_k^u$, where $\tau_k^l$ and $\tau_k^u$ denote the lower and upper bounds of the delay for the $k$th agent, respectively. 

By using (\ref{xi_sys}), the MAS with heterogeneous time-varying delays can be expressed into the following equation:
\beq\label{xdot}
    \dot x(t) = (I_m \otimes A) x(t)-\sum_{k=1}^m (E_kL\otimes BK) x(t-\tau_{k}(t)), 
\eeq
where $x(t)=\left[x_1^T(t), x_2^T(t),\cdots, x_m^T(t) \right]^T$ and $E_k\in\mathbb{R}^{m \times m}$ denotes a zero matrix except one at the $(k,k)$ entry.
Considering the consensus problem using the tree-type transformation \cite{sun2009consensus}, the difference between state variables $z_k(t)$ is expressed as follows:
\beqn
    z_k(t)=x_1(t)-x_{k+1}(t) ~~\forall k \in \{1,2,\cdots,m-1\}, 
\eeqn

Let $G = [1_{m-1} ~~-I_{m-1}]$ and $U=[0_{m-1}~~-I_{m-1}]^T$, which provides the following expression:
\beqn\baged
    \dot z(t)=& (I_{m-1}\otimes A)-\sum_{k=1}^m\bigg(GE_kLU\otimes BK\bigg)z(t-\tau_k(t)).
\eaged\eeqn

\begin{lemma}\label{lem1} (Alternative expression of generalized integral inequality on \cite{lee2020novel}) 
    Given $\zeta(t) \in \mathbb{R}^{mn}$, $N\in \mathbb{N}$, and a continuous and differentiable function $z(s) \in \mathbb{R}^n$ over the interval $[a, b]$, the following integral inequality holds for the positive definite matrix $R \in \mathbb{R}^{2n \times 2n}$, free matrices $Y_{i} \in R^{mn\times n} (i = 1, 2,\cdots,2N+1)$, and $\xi(s) =\bb \dot z^T (s) & z^T (s) \eb^T$:

     \small\beqn\baged
         -\int_{a}^{b}\xi^T(s) R \xi(s) ds \leq \zeta^T \bigg[(b-a)\bar Y \bar R^{-1}\bar Y^T+\textbf{Sym}\{\bar Y \bar M^T\}\bigg]\zeta,
     \eaged\eeqn\normalsize
     where  
     \beqn\baged
         &\bar Y = \bb Y_1&Y_2&\cdots & Y_{2N+1}\!&\!0\eb, \;
         \bar M = \bb M_1&M_2&\cdots& M_{2N+1}&0\eb,\\
         &\bar R=diag\{R,\;3R,\;\cdots,\;(2N+1)R\}, \\
         &M_{2j-1}\zeta(t)=	\int_a^b p_{j-1}(s)\dot z(s)ds, \\ 
         &M_{2j}\zeta(t) =\int_a^b p_{j-1}(s)z(s)ds,\\
         &p_j(s)=(-1)^{j}\sum_{i=0}^j \bigg[(-1)^i{j \choose i}{j+i \choose i}\bigg]\bigg({s-a\over b-a}\bigg)^i.
     \eaged\eeqn     
\end{lemma}



%% file: txts/sec3.tex
Throughout the paper, we use the compact notation
    \beqnset
         &\bar A:=I_{m-1}\otimes A, \quad \bar B_k:=GE_kLU\otimes (BK), \\
         & F=n(m-1), \quad f=\{1,2,\cdots,1+(3+3N)m\}, \\
         &\Tilde{\tau}_k:=\tau_k^u-\tau_k^l,\quad  \forall k\in\{1,\cdots,m\}. \\
    \eeqnset 
The closed-loop dynamics of the error state are 
\beq \label{eq:closed-loop dynamics}
    \dot z (t) = \bar Az(t)-\sum_{k=1}^m \bar B_k z(t-\tau_k(t)).
\eeq

Before deriving the main result, an augmented state is defined as follows:\\
$\zeta(t) = \left[\zeta_1^T(t),\zeta_2^T(t),\zeta_3^T(t),\zeta_4^T(t),\zeta_5^T(t)\right]^T$, where

\begin{align*}
	\zeta_1(t)&=[ z(t)^T \; z(t-\tau_1^l)^T \; \cdots \; z(t-\tau_m^l) \\ 
	&\hspace{12em} \; z(t-\tau_1^u) \; \cdots \; z(t-\tau_m^u) ], \\
    \zeta_2(t)&=[ z(t-\tau_1(t)) \; \cdots \; z(t-\tau_m(t)) ], \\
\end{align*}
\begin{align*}
	\zeta_3(t)\!&=\!\bb{1 \over \tau_1^l}\int_{t-\tau^l_1}^{t} z(s_1)ds_1\\\vdots \\{1 \over \tau_m^l}\int_{t-\tau_m^l}^{t} z(s_1)ds_1\\\vdots\\{(N-1)! \over (\tau^l_1)^N}\int_{t-\tau^l_1}^{t}\cdots \int_{r}^{t} z(s_1)ds_1\cdots ds_N\\\vdots\\{(N-1)! \over (\tau_m^l)^N}\int_{t-\tau_m^l}^{t}\cdots \int_{r}^{t} z(s_1)ds_1\cdots ds_N \eb, \\
	\zeta_4(t)&\!=\!\bb\! {1 \over (\tau_1(t)-\tau^l_1)}\int_{t-\tau_1(t)}^{t-\tau^l_1} z(s_1)ds_1 \\ \vdots \\  {1 \over (\tau_m(t)-\tau_m^l)}\int_{t-\tau_m(t)}^{t-\tau_m^l} z(s)ds\\ \vdots\\ {(N-1)\over (\tau_1(t)-\tau^l_1)^{N}}\int_{t-\tau^l_1(t)}^{t-\tau^l_1}\cdots \int_{r}^{t-\tau^l_1} z(s_1)ds_1 \cdots ds_N\\
	\vdots\\ {(N-1) \over (\tau_m(t)-\tau_m^l)^{N}}\int_{t-\tau_m(t)}^{t-\tau_m^l}\cdots \int_{r}^{t-\tau^l_m} z(s_1)ds_1 \cdots ds_N
	\!\eb, \\
	\zeta_5(t)&\!=\!\bb\! {1 \over (\tau_1^u-\tau_1(t))}\int_{t-\tau_1^u}^{t-\tau_1(t)} z(s_1)ds_1 \\ \vdots \\ {1 \over (\tau_m^u-\tau_m(t))}\int_{t-\tau_m^u}^{t-\tau_m(t)} z(s_1)ds_1\\\vdots\\{(N-1)! \over (\tau_1^u-\tau_1(t))^{N}}\int_{t-\tau_1^u}^{t-\tau_1(t)}\!\cdots\! \int_{r}^{t-\tau_1(t)} \! z(s_1)ds_1 \!\cdots\! ds_N\\\vdots\\ {(N-1)! \over (\tau_m^u-\tau_m(t))^{N}}\int_{t-\tau_m^u}^{t-\tau_m(t)}\!\cdots\! \int_{r}^{t-\tau_m(t)} \! z(s_1)ds_1 \!\cdots\! ds_N\!\eb\!.
\end{align*}
The augmented vector $\zeta(t)$ is constructed to capture the complete delay-related information required for the stability analysis. 
Specifically, $\zeta_3(t)$, $\zeta_4(t)$, and $\xi_5(t)$ consist of iterated integral terms over the sub-intervals $[t-\tau_k^l,t]$, $[t-\tau(k),t-\tau_k^l]$, and $[t-\tau(k),t-\tau_k^u]$, respectively. It arise from the application of Lemma~\ref{lem1} and contribute to reducing conservatism in the stability conditions.

Here, $e_f=\bb 0_{F \times F(f-1)} & I_{F} & 0_{F \times F(1+(3+3N)m-f)}\eb^T$ is defined such that $e_f^T\zeta(t)$ gives the $f$-th block of $\zeta(t)$. For example, $z(t)=e_1^T\zeta(t)$, $z(t-\tau_k^l)=e_{1+k}^T\zeta(t)$, and $z(t-\tau_k^u)=e_{1+m+k}^T\zeta(t)$.

The primary objective of the proposed framework is to characterize the maximum allowable delay bounds $\tau_k^l$ and $\tau_k^u$ that guarantee consensus for a given controller gain $K$. The following theorem presents the main result.

\begin{theorem} \label{thm}
	The multi-agent system in \eqref{xdot} with heterogeneous time delays $\tau_{k}(t)$ achieves the consensus if there exist $P_{k} \in \mathbb R^{3F\times3F}>0,~Q_{1k},~Q_{2k}\in \mathbb R^{F\times F}>0$, $R_{1k},~R_{2k}\in \mathbb R^{2F\times2F}>0$, $W_{1k},~ W_{2k},~W_{3k}\in \mathbb R^{F\times F}>0$, matrices $Y_{ikj}\in \mathbb R^{(1+(3+3N)m)F\times F}(i=1,2,3,~ j=1,\cdots,2N+1)$, such that the following LMIs hold for all $\tau_{k}(t)\in\{\tau_k^l,\tau_k^u\}$ and for all $k=\{1,\cdots,m\}$
	
	\beq \label{thm_lmi}\baged
	\barr{ccccc}
	\bar \Pi      & \bar X_{1}   & \bar X_{23} \\ 
	(*)         &-\overline {RW}_{1} & 0  \\ 
	(*)         & 0              & -\overline {RW}_{23} \earr&<0,	\\
	R_{1k}+\bar W_{1k}>0, \quad 	R_{2k}+\bar W_{2k}>0, \quad R_{2k}+\bar W_{3k}&>0,
	\eaged\eeq
	
	where
    \begin{align*}
        \bar \Pi\!=&\! \sum_{k=1}^m\bigg\{{\bf Sym}(\Pi_{2k}P_k\Pi_{1k}^T)+e_{1}(Q_{1k}+W_{1k})e_{1}^T\\
        +&e_{1+k}(Q_{2k}-Q_{1k}+W_{2k}-W_{1k})e_{1+k}^T\\
        +&e_{1+2m+k}(W_{3k}-W_{2k})e^T_{1+2m+k}\\
        -&e_{1+m+k}(Q_2+W_{3k})e_{1+m+k}^T
		+\Pi_{3}(\tau_k^l R_{1k}+\tilde\tau_kR_{2k})\Pi_{3}^T\\
		+&{\rm {\bf Sym}}(\bar Y_{1k}\bar M_{1k}^T+\bar Y_{2k}\bar M_{2k}^T+\bar Y_{3k}\bar M_{3k}^T) \bigg\}, \\
		\Pi_{1k}&\!=\![ e_1 ~~ \tau_k^le_{1+3m+k} ~~(\tau_{k}(t)-\tau_k^l)e_{1+(3+N)m+k}\\
        \Pi_{2k}&\!=\![ e_1\bar A^T-\sum_{k=1}^me_{1+2m+k}\bar B_{k}^T \\
		&\hspace{10em}  e_1-e_{1+k} e_{1+k}-e_{1+m+k}], \\
		\Pi_{3}&\!=\!\bb e_1\bar A^T-\sum_{k=1}^me_{1+2m+k}\bar B_{k}^T & e_1\eb, \\
		\bar A &\!=\! I_{m-1}\otimes A, ~~ \bar B_{k}\!=\!GE_kLU\otimes (BK), \\
		G &\!=\! [1_{m-1} ~~-I_{m-1}], ~U=[0_{m-1}~~-I_{m-1}]^T, \\
		\bar W_{1k} &=  \bb  0  &W_{1k} \\  W_{1k} & 0\eb,	\bar W_{2k} =  \bb  0  &W_{2k} \\  W_{2k} & 0\eb, \\
		\bar W_{3k} &=  \bb  0 & W_{3k} \\  W_{3k} & 0\eb, 
		\bar M_x \!=\ \bb M_1^x&M_2^x&\cdots& M_{2N+1}^x&0\eb, \\
		&\hspace{8em}+(\tau_k^u-\tau_{k}(t))\times e_{1+(3+2N)m+k}],\\
        M_{2j-1}^{x} \!&=\! \sum_{i=0}^{j-1}\gamma_{ij}\rho_{i}^{x},~~ M_{2j}^{1k} \!=\! \tau_k^l\sum_{i=0}^{j-1}\gamma_{ij}e_{1+(2+i)m+k},\\
		M_{2j}^{2k} &\!=\! (\tau_k(t)-\tau_k^l)\sum_{i=0}^{j-1}\gamma_{ij}e_{1+(2+i+N)m+k},\\
		M_{2j}^{3k} &\!=\! (\tau_k^u-\tau_k(t))\sum_{i=0}^{j-1}\gamma_{ij}e_{1+(2+i+2N)m+k}, \\
        \gamma_{ij} &\!=\!(-1)^{j-1}\bigg[(-1)^{i}{j-1 \choose i}{j+i-1 \choose i}\bigg], \\
		\rho_i^{1k} &\!=\! \begin{cases}	e_1-e_{1+k}, ~ i=0  \\ e_1-ie_{1+(2+i)m+k}, ~i \geq 1 \end{cases},\\
		\rho_{i}^{2k} &\!=\! \begin{cases}	e_{1+k}-e_{1+2m+k}, ~ i=0  \\ e_{1+k}-ie_{1+(2+i+N)m+k}, ~i \geq 1 \end{cases}, \\
		\rho_{i}^{3k} &\!=\! \begin{cases}	e_{1+2m+k}-e_{1+m+k}, ~ i=0  \\ e_{1+2m+k}-ie_{1+(2+i+2N)m+k}, ~i \geq 1 \end{cases},\\
        \bar X_{1} &\!=\! \bb \tau_1^l \bar Y_{11} &\cdots &\tau_m^l\bar Y_{1m}\eb, \\
		\bar X_{23} &\!=\! [ (\tau_{1}(t)-\tau_1^l)\bar Y_{21}+(\tau_1^u-\tau_{1}(t))\bar Y_{31} \\
		&\hspace{3em}\cdots ~~(\tau_{m}(t)-\tau_m^l)\bar Y_{2m}+(\tau_m^u-\tau_{m}(t))\bar Y_{3m}],\\
		\bar Y_{ik}&\!=\! \bb  Y_{ik1}&Y_{ik2}&\cdots   & Y_{ik(2N+1)}&  0\eb(i=1,2,3),\\	
        \overline {RW}_{1}&\!=\!diag\{\tilde\tau_1\Pi_{41},~ \cdots,~ \tilde\tau_m\Pi_{4m}\},\\
		\overline {RW}_{23}&\!=\!diag\{\tilde\tau_1\Pi_{51},~ \cdots,~ \tilde\tau_m\Pi_{5m}\},\\
        \Pi_{4k}&\!=\!\mathcal{I}_N\otimes(R_{1k}+W_{1k}), \\
		\Pi_{5k}&\!=\! (\tau_{k}(t)-\tau_k^l)\Pi_{6k}+(\tau_k^u-\tau_k(t))\Pi_{7k}, \\
    \end{align*}
	\begin{align*}
		\Pi_{6k}&\!=\!\mathcal{I}_N\otimes(R_{2k}+W_{2k}), \;
		\Pi_{7k}\!=\!\mathcal{I}_N\otimes(R_{2k}+W_{3k}),\\
		\mathcal{I}_N&\!=\!\text{diag}\{1,3,\cdots,(2N+1)\}.
	\end{align*}
	
	\begin{pf}
    \textbf{Step 1: LKF construction:} 
    Consider the Lyapunov-Krasovskii functional (LKF)
     \begin{align*}
        V(t)  &=V_1(t)+V_2(t) +V_3(t), \\
        V_1(t)&= \sum_{k=1}^m\eta_k^T(t) P_k\eta_k(t), \\
        V_2(t) &=\sum_{k=1}^m\int_{t-\tau_k^l}^{t} z^T(s)Q_{1k}z(s)ds,\\
        &\hspace{1em}+\sum_{k=1}^m\int_{t-\tau_k^u}^{t-\tau_k^l}z^T(s)Q_{2k}z(s)ds,\\
        V_3(t)  &=\sum_{k=1}^m\int_{-\tau_k^l}^0 \int_{t+s}^{t} \xi^T(r)R_{1k}\xi(r)dr ds \\
        &\hspace{1em}+\sum_{k=1}^m\int_{-\tau_k^u}^{-\tau_k^l} \int_{t+s}^{t} \xi^T(r)R_{2k}\xi(r)dr ds,
    \end{align*}

		where $\eta_k^T(t)=\barr{ccc}z^T(t)&\int_{t-\tau_k^l}^{t} z^T(s)ds&\int_{t-\tau_k^u}^{t-\tau_k^l} z^T(s)ds\earr$ and $\xi^T(t)= \barr{cc}\dot z^T(t)&z^T(t)\earr$.

    \textbf{Step 2: Derivative of $V_1(t)$:}
    The third component of $\eta_k(t)$ can be split over the interval at $t-\tau_k(t)$ as
    \beqnset
        &(\tau_k(t)-\tau_k^l)\underbrace{\frac{1}{\tau_k(t)-\tau_k^l}\int_{t-\tau_k(t)}^{t-\tau_k^l} z \, ds}_{e_{1+(3+N)m+k}^T\zeta(t)} \\
        &+(\tau_k^u-\tau_k(t))\underbrace{\frac{1}{\tau_k^u-\tau_k(t)} \int_{t-\tau_k^u}^{t-\tau_k(t)} z \, ds}_{e_{1+(3+2N)m+k}^T\zeta(t)},
    \eeqnset
    where the normalized integrals correspond to the first entries of $\zeta_4(t)$ and $\zeta_5(t)$, respectively.
    Together with $z(t)=e_1^T\zeta(t)$ and $\tau_k^l \cdot \frac{1}{\tau_k^l}\int_{t-\tau_k^l}^t z \, ds=\tau_k^l e_{1+2m+k}^T\zeta(t)$, it follows that $\eta_k(t)=\Pi_{1k}^T\zeta(t)$. Differentiating $\eta_k(t)$ via the Leibniz rule gives
    \beqn
        \dot\eta_k(t)=\bb \dot z(t) \\ z(t)-z^T(t-\tau_k^l) \\ z(t-\tau_k^l)-z(t-\tau_k^u)\eb.
    \eeqn
    Substituting the closed-loop dynamics \eqref{eq:closed-loop dynamics} and using $z(t)=e_1^T\zeta(t)$, $z(t-\tau_k^l)=e_{1+k}^T\zeta(t)$, $z(t-\tau_k^u)=e_{1+m+k}^T\zeta(t)$, $z(t-\tau_k(t))=e_{1+2m+k}^T\zeta(t)$, leads to $\dot\eta_k(t)=\Pi_{2k}^T\zeta(t)$.

    Taking the time derivative of $V_1(t)=\sum_{k=1}^m \eta_k^T(t)P_k\eta_k(t)$ gives
    \beq\label{dotv1}
        \dot V_1(t) = \sum_{k=1}^m 2\eta_k^T(t)P_k\dot \eta_k(t)=2\zeta^T(t)\left[\sum_{k=1}^m \Pi_{2k}P_k\Pi_{1k}^T\right]\zeta(t).
    \eeq

    \textbf{Step 3: Derivative of $V_2(t)$:}
    The time derivative of $V_2(t)$ is given by
    \begin{align}
        &\dot V_2(t)=\sum_{k=1}^m [z^T(t)Q_{1k}z(t)-z^T(t-\tau_k^lQ_{1k}z(t-\tau_k^l)], \nonumber \\
                    &+\sum_{k=1}^m [z^T(t-\tau_k^l)Q_{2k}z(t-\tau_k^l)-z^T(t-\tau_k^u)Q_{2k}z(t-\tau_k^u)],  \nonumber \\
                &= \sum_{k=1}^m\zeta^T(t)\bigg\{e_1Q_{1k}e_1^T+e_{1+k}(Q_{2k}-Q_{1k})e_{1+k}^T\\
			&\hspace{12em}-e_{1+m+k}Q_{2k}e_{1+m+k}^T\bigg\}\zeta(t) \nonumber .
    \end{align}\label{dotv2}

    \textbf{Step 4: Derivative of $V_3(t)$:}
    \beqset\label{v3}
        &\dot V_3(t) =\sum_{k=1}^m \xi^T(t)(\tau_k^lR_{1k}+\tilde \tau_kR_{2k})\xi(t) \\
        &- \sum_{k=1}^m\int_{t-\tau_k^l}^t \xi^T(s)R_{1k}\xi(s)ds -\sum_{k=1}^m \int_{t-\tau_k^u}^{t-\tau_k^l} \xi^T(s)R_{2k}\xi(s) ds.
    \eeqset

    The last term of \eqref{v3} can be separated into two parts.
    \begin{align}\label{v3_term}
			 \int_{t-\tau_{k}(t)}^{t-\tau_k^l} \xi^T(s) R_{2k}\xi(s) ds+\int_{t-\tau_k^u}^{t-\tau_{k}(t)} \xi^T(s)R_{2k}\xi(s) ds.
    \end{align}

    To incorporate the weighting matrices $W_{1K}$, $W_{2k}$, $W_{3k}$, the following zero equalities are considered.
    \begin{align}
			0&=\sum_{k=1}^m\bigg(z^T(t)W_{1k}z(t)-z^T(t-\tau_k^l)W_{1k}z(t-\tau_k^l)  \nonumber \\
			&\hspace{1em}-2\int_{t-\tau_k^l}^tz^T(s)W_{1k}\dot z(s)ds\bigg),  \label{zero_1} \\
			0&=\sum_{k=1}^m \bigg(z^T(t-\tau_k^l)W_{2k}z(t-\tau_k^l)-z^T(t-\tau_k(t))W_{2k} \nonumber \\
			&\hspace{1em}\times z(t-\tau_k(t))-2\int_{t-\tau_k(t)}^{t-\tau_k^l}z^T(s)W_{2k}\dot z(s)ds\bigg),  \label{zero_2}\\
			0&=\sum_{k=1}^m \bigg(z^T(t-\tau_k(t))W_{3k}z(t-\tau_k(t)) \label{zero_3} \\
			-&z^T(t-\tau_k^u)W_{3k}z(t-\tau_k^u) 
			-2\int_{t-\tau_k^u}^{t-\tau_k(t)}z^T(s)W_{3k}\dot z(s)ds\bigg). \nonumber
		\end{align}
        By summing the zero equalities \eqref{zero_1}--\eqref{zero_3}, $\dot V_3(t)$ can be expressed as
        \begin{align}\label{dotv3}
			&\dot V_3(t)=\sum_{k=1}^m\bigg\{\xi^T(t)(\tau_k^l R_{1k}+\tilde\tau_k R_{2k})\xi(t)+ z^T(t)W_{1k}z(t) \nonumber\\
			&\hspace{0em}+z^T(t-\tau_k^l)(W_{2k}-W_{1k})z(t-\tau_k^l) \nonumber\\
			&\hspace{0em}+z^T(t-\tau_k(t))(W_{3k}-W_{2k})z(t-\tau_k(t)) \nonumber \\
			&\hspace{0em}-z^T(t-\tau_k^u)W_{3k}z(t-\tau_k^u)  \nonumber \\
			&\hspace{0em}-\int_{t-\tau_k^l}^{t} \xi^T(s)(R_{1k}+\bar W_{1k})\xi(s) ds\nonumber
			\hspace{0em}- \int_{t-\tau_{k}(t)}^{t-\tau_k^l} \xi^T(s) \cdot \\
            &(R_{2k}+\bar W_{2k})\xi(s) -\int_{t-\tau_k^u}^{t-\tau_{k}(t)} \xi^T(s)(R_{2k}+\bar W_{3k})\xi(s) ds\bigg\}, 
		\end{align}
        where $2z^T(t)W\dot z(t)=\xi^T(t)\bar W \xi(t)$ with $\bar W=\bb 0 & W \\ W & 0\eb$.

        Based on the Lemma \ref{lem1}, $\dot V_3(t)$ can be bounded as follows
		\begin{align}\label{dotv3_ineq}
			\dot V_3(t)&\leq\sum_{k=1}^m\bigg\{\xi^T(t)(\tau_k^lR_{1k}+\tilde\tau_kR_{2k})\xi(t)+z^T(t)W_{1k}z(t) \nonumber  \\
			+&z^T(t-\tau_k^l)(W_{2k}-W_{1k})z(t-\tau_k^l)\nonumber\\
			+&z^T(t-\tau_k(t))(W_{3k}-W_{2k})z(t-\tau_k(t)) \nonumber \\
			-&z^T(t-\tau_k^u)W_{3k}z(t-\tau_k^u) 
			+\zeta^T(t)\{\tau_k^l\bar Y_{1k}{RW}_{1k}^{-1}\bar Y_{1k}^T\nonumber\\
			+&(\tau_{k}(t)-\tau_k^l)\bar Y_{2k}{RW}_{2k}^{-1}\bar Y_{2k}^T
			+(\tau_k^u-\tau_{k}(t))\bar Y_{3k}{RW}_{3k}^{-1}\bar Y_{3k}^T \nonumber \\
			+&{\rm {\bf Sym}}(\bar Y_{1k}\bar M_{1k}^T 
			+\bar Y_{2k}\bar M_{2k}^T+\bar Y_{3k}\bar M_{3k}^T )\}\zeta(t)\bigg\},
		\end{align}
        where $RW_{1k} = \mathcal{I}_N\otimes(R_{1k}+W_{1k})$, 
        $RW_{2k} = \mathcal{I}_N\otimes(R_{2k}+W_{2k})$, 
        $RW_{3k} = \mathcal{I}_N\otimes(R_{2k}+W_{3k})$, 
        which are positive definite by \eqref{thm_lmi}.

    \textbf{Derivation of $\dot V(t)$}    
    Summing \eqref{dotv1}, \eqref{dotv2}, and \eqref{dotv3_ineq}, and collecting all quadratic terms in $\zeta(t)$, an upper bound of the LKF is obtained as
    \begin{align}\label{dotv}
        \dot{V}(t)&<\zeta^T(t)\bigg\{\bar \Pi + \sum_{k=1}^m\bigg(\tau_k^l \bar Y_{1k}{RW}_{1k}^{-1}\bar Y_{1k}^T + (\tau_{k}(t)-\tau_k^l) \nonumber \\
         \times \bar Y_{2k}& {RW}_{2k}^{-1}\bar Y_{2k}^T +(\tau_k^u-\tau_{k}(t))\bar Y_{3k} {RW}_{3k}^{-1}\bar Y_{3k}^T\bigg)\bigg\}\zeta(t),
    \end{align}
    where $\bar \Pi$ collects all delay-independent quadratic terms.
    
    Here, $\Pi_3=\bb e_1\bar A^T-\sum_{k=1}^m e_{1+2m+k}\bar B_{k}^T & e_1\eb$ maps $\xi(t)$ into the $\zeta$-coordinate, and the right-hand side of \eqref{dotv} contains three terms of the form $c_k(t)\bar Y_{ik}M_k^{-1}\bar Y_{ik}^T$.
    
    By applying the Schur complement lemma,
    \beqn
        \bar \Pi+c\bar YM^{-1}\bar Y^T<0 \longleftrightarrow \bb \bar \Pi & c\bar Y \\ \ast & -cM\eb< 0,
    \eeqn
    \eqref{dotv} can be rewritten as
    \beq\label{pf_thm}
        \dot V(t)<\zeta^T(t)\Xi(\tau_k(t))\zeta(t),
    \eeq
    where $\Xi(\tau_k(t))$ is given by
    \beq\label{thm_xi}
        \Xi(\tau_k(t))=\bb \bar\Pi & \bar X_1 & \bar X_2 & \bar X_3 \\
            \ast & -\overline{RW}_1 & 0 & 0 \\
            \ast & 0 & \overline{RW}_2 & 0 \\
            \ast & \ast & \ast & \overline{RW}_3 \eb,
    \eeq
    \beqnset
        &\bar X_2 = \bb(\tau_1(t)-\tau_1^l)\bar Y_{21} \quad & \cdots & \quad (\tau_m(t)-\tau_m^l)\bar Y_{2m} \eb, \\
        &\bar X_3= \bb(\tau_1^u-\tau_1(t))\bar Y_{31} \quad & \cdots & \quad (\tau_m^u-\tau_m(t))\bar Y_{3m} \eb, \\
        &\overline{RW}_2= -\text{diag}\left\{ (\tau_1(t)-\tau_1^l)\Pi_{61}, \cdots, (\tau_m(t)-\tau_m^l)\Pi_{6m} \right\}, \\
        &\overline{RW}_3= -\text{diag}\left\{ (\tau_1^u-\tau_1(t))\Pi_{71}, \cdots, (\tau_m^u-\tau_m(t))\Pi_{7m} \right\}.
    \eeqnset
    
    Since $\Xi(\tau_k(t))$ is affine in each $\tau_k(t)$ over the interval $[\tau_k^l,\tau_k^u]$, it can be expressed as a convex combination of its values at the end points. Therefore, $\Xi(\tau_k(t))<0$ for all $\tau_k(t)\in [\tau_k^l, \tau_k^u]$ if and only if 
    \beqn
        \Xi(\tau_k^l)<0 \quad \text{and} \quad \Xi(\tau_k^u)<0,
    \eeqn
    which corresponds exactly to the LMI conditions in \eqref{thm_lmi}, ensuring $\dot{V}(t)<0$. This completes the proof. $\blacksquare$ 
	\end{pf}
\end{theorem}

\begin{remark}  
    A higher approximation order $N$ reduces conservatism by incorporating more information, but increases the LMI dimension by $\mathcal{O}(m^2n)$. This $N$ should be selected by balancing conservatism reduction against the available computational budget.
\end{remark}

\begin{remark}
    The zero equalities \eqref{zero_1}--\eqref{zero_3} introduce independent weighting matrices $W_{1k}$, $W_{2k}$, $W_{3k}$ for each sub-interval, enabling interval-wise tuning of the stability conditions. Furthermore, the LMI conditions do not require $W_{ik}$ to be positive definite, thereby enlarging the feasible search space and yielding less conservative results. 
   Unlike existing works that impose a common delay bound across all agents, the proposed criterion assigns 
    individual bounds $\tau_k^l$ and $\tau_k^u$ per agent, which reduces conservatism.
\end{remark}

Eliminating the zero equalities \eqref{zero_1}-\eqref{zero_3} in the LKF in Theorem \ref{thm}, following corollary can be obtained:

\begin{corollary} \label{cor}
	The multi-agent system in \eqref{xdot} with heterogeneous time delays $\tau_{k}(t)\in[\tau_k^1, \tau_k^u]$ achieves the consensus if there exist $P_{k} \in \mathbb R^{3F\times3F}>0,~Q_{1k},~Q_{2k}\in \mathbb R^{F\times F}>0$, $R_{1k},~R_{2k}\in \mathbb R^{2F\times2F}>0$, matrices $Y_{ikj}\in \mathbb R^{(1+(3+3N)m)F\times F}(i=1,2,3,~ j=1,\cdots,2N+1)$, such that the following LMIs hold for all $\tau_{k}(t)\in\{\tau_k^1, \tau_k^u\}$ and for all $k=\{1,\cdots,m\}$
	\beq\label{cor_lmi}
	\barr{c|c|ccc} 
	\tilde \Pi      & \bar X_{1}   & \bar X_{23} \\ \hline
	(*)         &-\overline{R}_{1} & 0                      \\ \hline
	(*)         & 0              & -\overline {R}_{2}  \earr,	
	\eeq
	where all parameters are the same as those defined in Theorem \ref{thm}, except for the following:
	\beqn\baged
	\overline {R}_{1} &= diag\{\tau_1^l{R}_{11},\cdots,\tau_m^l{R}_{1m}\},\\
	\overline {R}_{2} &= diag\{\tilde\tau_1{R}_{21},\cdots, \tilde\tau_m{R}_{2m}\},\\
	\tilde \Pi &=\sum_{k=1}^m\bigg\{{\bf Sym}(\Pi_{2k}P_k\Pi_{1k}^T)+e_1Q_{1k}e_1^T \\
	&\hspace{2em}+e_{1+k}(Q_{2k}-Q_{1k})e_{1+k}^T-e_{1+m+k}Q_{2k}e_{1+m+k}^T\\
	&\hspace{2em}+\Pi_{3}(\tau_k^l R_{1k}+\tilde\tau_kR_{2k})\Pi_{3}^T\\
	&\hspace{2em}+{\rm {\bf Sym}}(\bar Y_{1k}\bar M_{1k}^T+\bar Y_{2k}\bar M_{2k}^T+\bar Y_{3k}\bar M_{3k}^T)\bigg\}.
	\eaged\eeqn
\end{corollary}

%% file: txts/sec4.tex
This section details two examples. Example 1 illustrates the efficiency of the proposed theorem by comparing it with early work. 
Example 2 reveals that the proposed theorem can reach a consensus under heterogeneous time-varying delays.
\begin{table}[t!]
    \newcolumntype{P}[1]{>{\centering\arraybackslash}p{#1}}
    \renewcommand{\arraystretch}{1.2}
	\caption{Maximum allowable delay $\tau^u$ when $\tau^l=0$ for example 1}
	\label{table1}
	\begin{center}
		\begin{tabular}
        {|P{0.3\columnwidth}|P{0.1\columnwidth}|P{0.1\columnwidth}|P{0.1\columnwidth}|P{0.1\columnwidth}|}
			\hline
			$\tau^l$ & 0 & 0.1  & 0.2 &0.3  \\	\hline
			\cite{savino2015conditions}& 0.416	& 0.436& 0.456  & 0.478 \\	\hline
			Cor \ref{cor}(N=1)&0.521 &0.655 & 0.775  &0.870 \\	\hline
			Cor \ref{cor}(N=2)& 0.521& 0.706&0.829 &0.922\\	\hline
			Thm\ref{thm}(N=1)&0.522 &0.656  &0.780 &0.876 \\ \hline
			Thm \ref{thm}(N=2)&0.522 & 0.730 & 0.863&0.957 \\	\hline			
		\end{tabular}
	\end{center}
\end{table}

\begin{example}
    For a practical example, consider the case of a second-order robot model used in \cite{michael2014experimental} with heterogeneous time-varying delays. The system description is as follows:
    \beq \label{ex1}
         \ddot p_k(t)+d\dot p_k(t) +sp_k(t) = u_k(t-\tau_{k}(t)), 
    \eeq
    where $p_k$ is the position state of $k$th agent. The parameters $d$ and $s$ are the damping and spring constant. 
    By setting $x_k(t) = col\{p_k(t), ~\dot p_k(t)\}$, entire MAS dynamics \eqref{xdot} is obtained, whose parameters are as follows:
    \beqn\baged
        A &= \bb 0 & 1 \\ -s & -d\eb, \quad  B= \bb 0 & 0 \\ 1 & 0\eb, \quad  K = \bb k & 0 \\ 0 & 0\eb,
    \eaged\eeqn
     where $d=1$, $s=0$, and $k=1$. These parameters and topology are selected in the paper \cite{savino2015conditions} for comparison. 
     Table \ref{table1} details the maximum allowable delay of \cite{savino2015conditions}, the proposed corollary, and theorem. 
     This demonstrates that the proposed criteria, reflecting the delay of each agent, have the larger maximum allowable delay than previous study considering the common bound of all agents. 
     By comparing the values of the Corollary \ref{cor} and Theorem \ref{thm}, the efficacy of zero equality condition is demonstrated. Furthermore, a larger $N$ guarantees a greater maximum allowable delay. Therefore, the proposed generalized consensus analysis allows users to balance computational complexity and delay tolerance according to their design objectives.
\end{example}

\begin{figure}
    \centering
    \includegraphics[width=\linewidth]{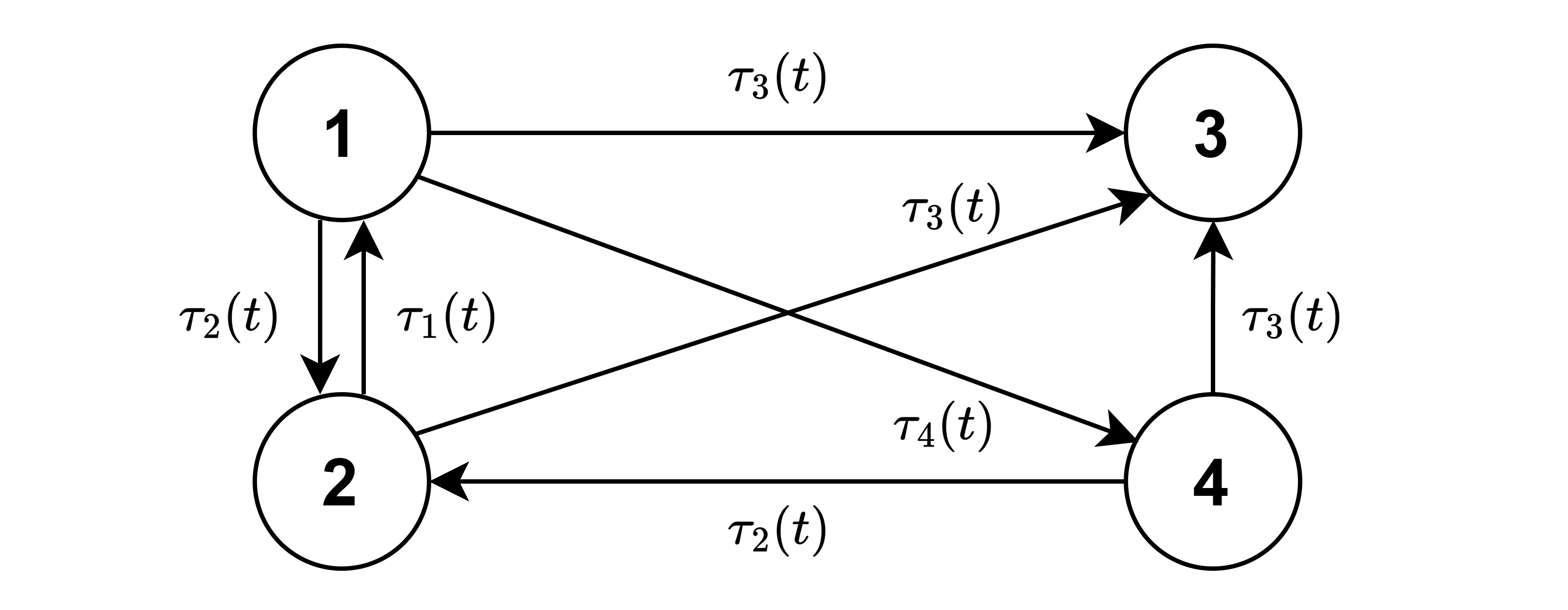}
    \caption{Communication topology with heterogeneous time-varying delay for Example 2}
    \label{fig1}
\end{figure}

\begin{figure}
    \centering
    \includegraphics[width=0.8\linewidth]{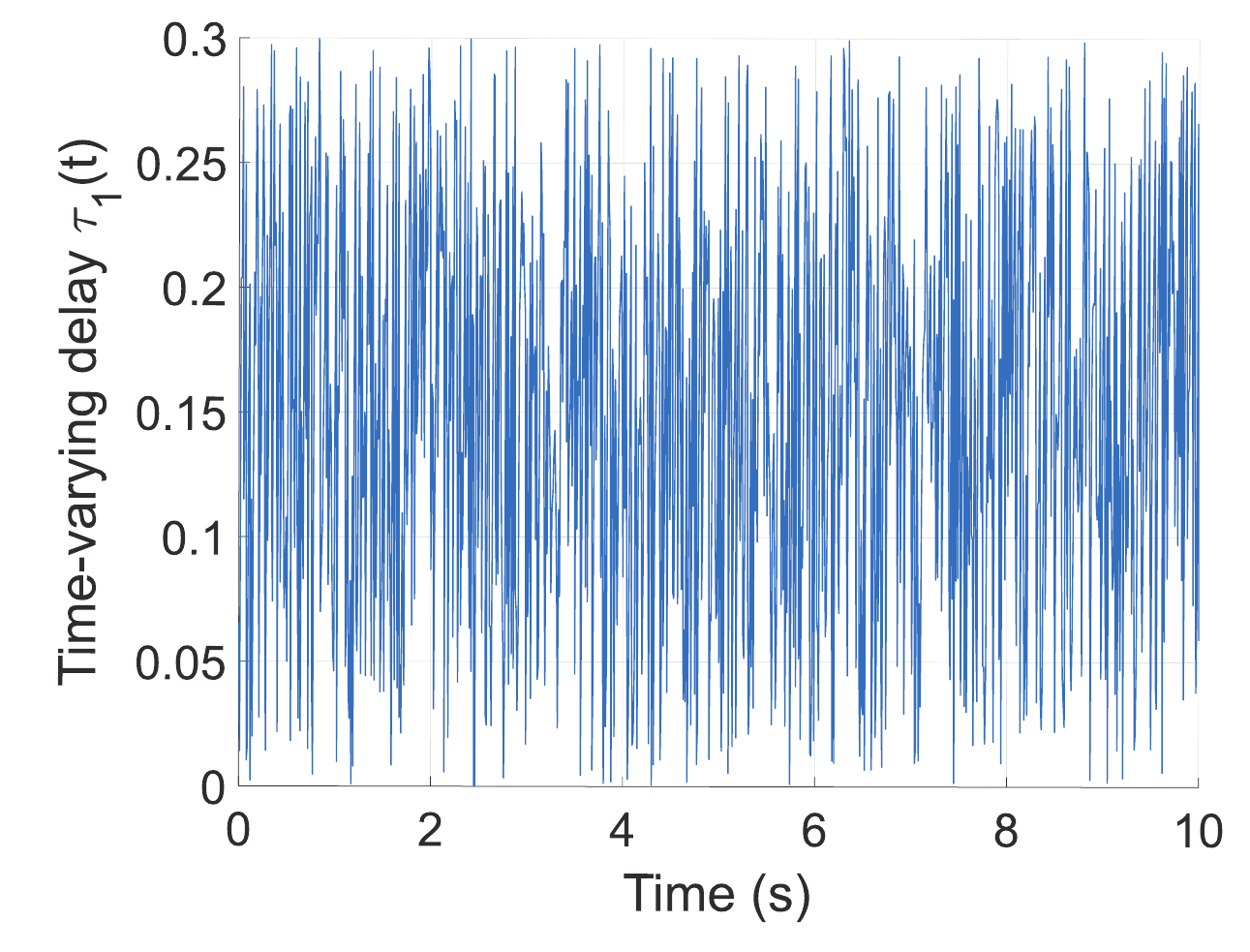}
    \caption{Time-varying delay $\tau_1(t)$}
    \label{fig2}
\end{figure}

\begin{figure}
    \centering
    \includegraphics[width=\linewidth]{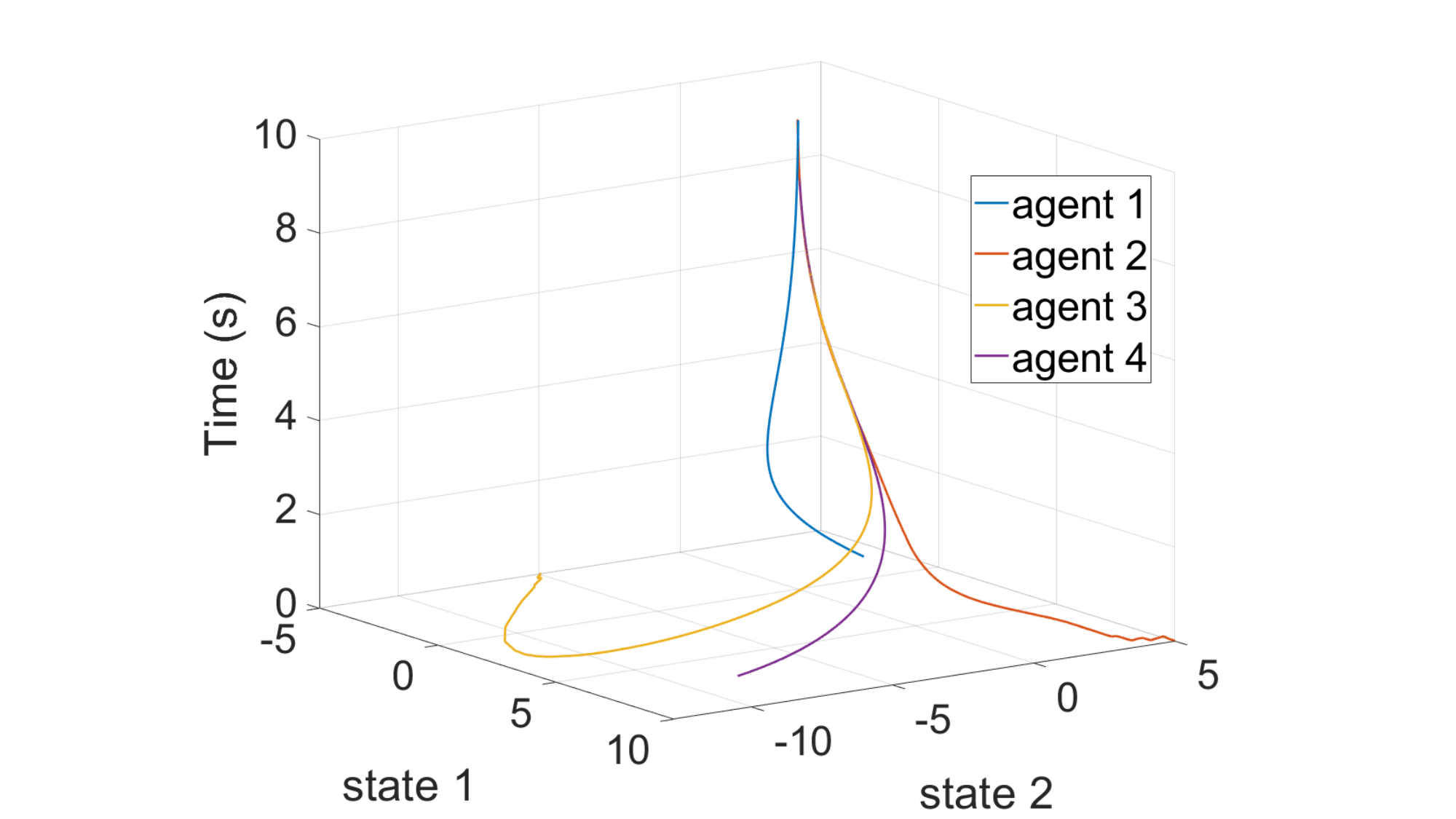}
    \caption{Trajectories of states for Example 2}
    \label{fig3}
\end{figure}

\begin{example}
    This example addresses the consensus problem for MAS with heterogeneous time-varying delays, with each agent having different bounds. 
	Consider the MAS \eqref{xdot} with following parameters: 
    \beqn\baged
        A&=\bb -0.5 & -0.6\\0.2 & -0.5 \eb, \; B= \bb 0.8 & -0.1 \\ -1.2 & 0.3 \eb, \;
        K=\bb -3 & 0.5 \\ 0.5 & -2 \eb, \\
        \tau_1^u&=0.3, ~\tau_2^u=0.2, ~\tau_3^u=0.2, ~\tau_4^u=0.1, ~\tau_i^l=0 ~\forall{i}.
    \eaged\eeqn

    The initial values are $x_{1}(0)=[-2~ 4]^T$, $x_{2}(0)=[10~ 5]^T$, $x_{3}(0)=[-5~ -5]^T$, and $x_{4}(0)=[7~ -8]^T$.
    Figure~\ref{fig1} shows the communication topology of the system 
    with heterogeneous time-varying delays for each agent. 
    Figure~\ref{fig2} depicts the time-varying delay of agent $1$, 
    while Fig.~\ref{fig3} presents the state trajectories of the MAS. 
    All agents asymptotically converge to a common consensus value of approximately $[0.028,~-0.024]$, confirming that the proposed criterion successfully guarantees consensus under heterogeneous time-varying delays.
\end{example}

%% file: reference.bib
@article{liu2021pseudo,
	title={Pseudo-predictor feedback control for multiagent systems with both state and input delays},
	author={Liu, Qingsong},
	journal={IEEE/CAA Journal of Automatica Sinica},
	volume={8},
	number={11},
	pages={1827--1836},
	year={2021},
	publisher={IEEE}
}

@article{you2023control,
	title={Control for Nonlinear Fuzzy Time-Delay Multiagent Systems: Two Kinds of Distributed Saturation-Constraint Impulsive Approach},
	author={You, Le and Jiang, Xiaowei and Li, Bo and Zhang, Xianhe and Yan, Huaicheng and Huang, Tingwen},
	journal={IEEE Transactions on Fuzzy Systems},
	volume={31},
	number={8},
	pages={2861--2870},
	year={2023},
	publisher={IEEE}
}

@article{zhang2023leader,
	title={Leader-following consensus of multi-agent systems with time delays by fully distributed protocols},
	author={Zhang, Kai and Zhou, Bin and Duan, Guang-Ren},
	journal={Systems \& Control Letters},
	volume={178},
	pages={105582},
	year={2023},
	publisher={Elsevier}
}

@article{tang2023edge,
	title={Edge-based self-triggering impulsive consensus on nonlinear multi-agent systems with proportional delay},
	author={Tang, Ze and Wang, Kun-Peng and Feng, Jianwen and Park, Ju H},
	journal={IEEE Transactions on Automation Science and Engineering},
	year={2023},
	publisher={IEEE}
}

@article{li2023delay,
	title={A delay classification-based approach to distributed consensus of nonlinear time-delay multiagent systems},
	author={Li, Kuo and Ahn, Choon Ki and Zheng, Wei Xing and Hua, Chang-Chun},
	journal={IEEE Transactions on Automatic Control},
	volume={68},
	number={11},
	pages={6990--6997},
	year={2023},
	publisher={IEEE}
}

@article{deng2020dynamic,
	title={A dynamic periodic event-triggered approach to consensus of heterogeneous linear multiagent systems with time-varying communication delays},
	author={Deng, Chao and Che, Wei-Wei and Wu, Zheng-Guang},
	journal={IEEE transactions on cybernetics},
	volume={51},
	number={4},
	pages={1812--1821},
	year={2020},
	publisher={IEEE}
}

@article{afrazi2025enhanced,
	title={Enhanced Density-driven Control of Multi-agent UAV Systems for Efficient Victim Detection in Large-scale Disaster Scenarios},
	author={Afrazi, Mohammad and Seo, Sungjun and Lee, Kooktae},
	journal={International Journal of Control, Automation and Systems},
	volume={23},
	number={9},
	pages={2728--2742},
	year={2025},
	publisher={Springer}
}

@article{lee2025graph,
	title={Graph Network Centralization via Asymmetric Edge Weight Allocation: Laplacian Conditioning and Multi-UAV System Application},
	author={Lee, Hye Jin and Na, Hyeon-Woo and Park, PooGyeon},
	journal={IEEE Transactions on Network Science and Engineering},
	year={2025},
	publisher={IEEE}
}

@article{hua2025multi,
	title={Multi-agent reinforcement learning for connected and automated vehicles control: Recent advancements and future prospects},
	author={Hua, Min and Qi, Xinda and Chen, Dong and Jiang, Kun and Liu, Zemin Eitan and Sun, Hongyu and Zhou, Quan and Xu, Hongming},
	journal={IEEE Transactions on Automation Science and Engineering},
	year={2025},
	publisher={IEEE}
}

@article{gao2022safety,
	title={Safety-critical model-free control for multi-target tracking of USVs with collision avoidance},
	author={Gao, Shengnan and Peng, Zhouhua and Wang, Haoliang and Liu, Lu and Wang, Dan},
	journal={IEEE/CAA Journal of Automatica Sinica},
	volume={9},
	number={7},
	pages={1323--1326},
	year={2022},
	publisher={IEEE}
}

@article{li2024multi,
	title={Multi-mode filter target tracking method for mobile robot using multi-agent reinforcement learning},
	author={Li, Xiaofeng and Ren, Jie and Li, Yunbo},
	journal={Engineering Applications of Artificial Intelligence},
	volume={127},
	pages={107398},
	year={2024},
	publisher={Elsevier}
}

@article{feng2024adaptive,
	title={Adaptive event-triggered time-varying output group formation containment control of heterogeneous multiagent systems},
	author={Feng, Lihong and Huang, Bonan and Sun, Jiayue and Sun, Qiuye and Xie, Xiangpeng},
	journal={IEEE/CAA Journal of Automatica Sinica},
	volume={11},
	number={6},
	pages={1398--1409},
	year={2024},
	publisher={IEEE}
}

@article{yang2024nonsingular,
	title={Nonsingular fixed-time consensus tracking for heterogeneous multi-agent systems with external disturbances and actuator faults},
	author={Yang, Pu and Ding, Yu and Feng, Ke-Jia and Shen, Zi-Wei},
	journal={International Journal of Control, Automation and Systems},
	volume={22},
	number={3},
	pages={840--850},
	year={2024},
	publisher={Springer}
}

@article{xue2023input,
	title={Input-to-state consensus of nonlinear positive multi-agent systems under state feedback and impulsive control},
	author={Xue, Yi-Qing and Zhao, Ping},
	journal={International Journal of Control, Automation and Systems},
	volume={21},
	number={7},
	pages={2099--2111},
	year={2023},
	publisher={Springer}
}

@article{xu2024event,
	title={Event-triggered leader-following consensus control of multiagent systems against DoS attacks},
	author={Xu, Fei and Ruan, Xiaoli and Pan, Xiong},
	journal={International Journal of Control, Automation and Systems},
	volume={22},
	number={11},
	pages={3424--3433},
	year={2024},
	publisher={Springer}
}

@article{ye2025adaptive,
	title={Adaptive descriptor sliding-mode observer-based dynamic event-triggered consensus of multiagent systems against actuator and sensor faults},
	author={Ye, Zhengyu and Jiang, Bin and Yu, Ziquan and Cheng, Yuehua},
	journal={IEEE Transactions on Cybernetics},
	year={2025},
	publisher={IEEE}
}

@article{hu2024event,
	title={Event-triggered consensus of multiagent systems with prescribed performance},
	author={Hu, Wenfeng and Hou, Yahui and Chen, Zhiyong and Yang, Chunhua and Gui, Weihua},
	journal={IEEE Transactions on Automatic Control},
	volume={69},
	number={8},
	pages={5462--5469},
	year={2024},
	publisher={IEEE}
}

@article{wen2024output,
	title={Output feedback consensus for high-order stochastic multi-agent systems with unknown time-varying delays},
	author={Wen, Baoyu and Huang, Jiangshuai},
	journal={International Journal of Control, Automation and Systems},
	volume={22},
	number={9},
	pages={2823--2832},
	year={2024},
	publisher={Springer}
}

@article{vu2024matrix,
	title={Matrix-scaled consensus of network systems with uniform time-delays},
	author={Vu, Hoang Huy and Nguyen, Quyen Ngoc and Trinh, Minh Hoang and Van Pham, Tuynh},
	journal={International Journal of Control, Automation and Systems},
	volume={22},
	number={9},
	pages={2783--2791},
	year={2024},
	publisher={Springer}
}

@article{savino2015conditions,
	title={Conditions for consensus of multi-agent systems with time-delays and uncertain switching topology},
	author={Savino, Heitor J and dos Santos, Carlos RP and Souza, Fernando O and Pimenta, Luciano CA and De Oliveira, Maur{\'\i}cio and Palhares, Reinaldo M},
	journal={IEEE Transactions on Industrial Electronics},
	volume={63},
	number={2},
	pages={1258--1267},
	year={2015},
	publisher={IEEE}
}

@article{wu2024sampled,
	title={A sampled-data-based secure control approach for networked control systems under random DoS attacks},
	author={Wu, Jiancun and Peng, Chen and Zhang, Jin and Tian, Engang},
	journal={IEEE Transactions on Cybernetics},
	volume={54},
	number={8},
	pages={4841--4851},
	year={2024},
	publisher={IEEE}
}

@article{zeng2023weighted,
	title={Weighted average consensus of second-order multi-agent systems with various intelligence levels via sampled-data control},
	author={Zeng, Jinhui and Wang, Pin and Lan, Zheng and Luo, Xueming and Li, Jianghong and Hu, Jiaxi},
	journal={International Journal of Control, Automation and Systems},
	volume={21},
	number={5},
	pages={1560--1569},
	year={2023},
	publisher={Springer}
}

@article{zhao2024observer,
	title={Observer-based sampled-data adaptive tracking control for heterogeneous nonlinear multi-agent systems under denial-of-service attacks},
	author={Zhao, Ning and Zhang, Huiyan and Shi, Peng},
	journal={IEEE Transactions on Automation Science and Engineering},
	volume={22},
	pages={4771--4779},
	year={2024},
	publisher={IEEE}
}

@article{lee2020novel,
	title={A novel generalized integral inequality based on free matrices for stability analysis of time-varying delay systems},
	author={Lee, Jun Hui and Park, In Seok and Park, Poogyeon},
	journal={IEEE Access},
	volume={8},
	pages={179772--179777},
	year={2020},
	publisher={IEEE}
}

@article{sun2009consensus,
	title={Consensus of multi-agent systems in directed networks with nonuniform time-varying delays},
	author={Sun, Yuan Gong and Wang, Long},
	journal={IEEE Transactions on Automatic Control},
	volume={54},
	number={7},
	pages={1607--1613},
	year={2009},
	publisher={IEEE}
}

@inproceedings{michael2014experimental,
	title={An experimental study of time scales and stability in networked multi-robot systems},
	author={Michael, Nathan and Schwager, Mac and Kumar, Vijay and Rus, Daniela},
	booktitle={Experimental robotics},
	pages={631--643},
	year={2014},
	organization={Springer}
}

@article{jiang2022multi,
	title={Multi-agent consensus with heterogeneous time-varying input and communication delays in digraphs},
	author={Jiang, Wei and Liu, Kun and Charalambous, Themistoklis},
	journal={Automatica},
	volume={135},
	pages={109950},
	year={2022},
	publisher={Elsevier}
}

@article{li2024hybrid,
	title={A hybrid time/event-triggered interaction framework for multi-agent consensus with relative measurements},
	author={Li, Xianwei and Tang, Yang and Zou, Yuanyuan and Li, Shaoyuan and Zheng, Wei Xing},
	journal={Automatica},
	volume={159},
	pages={111369},
	year={2024},
	publisher={Elsevier}
}
